\documentclass[a4paper,11pt]{amsart}
\usepackage[letterpaper, margin=1.2in]{geometry}
\usepackage{latexsym}
\usepackage{amssymb}
\usepackage{amsmath}
\usepackage{amsfonts}
\usepackage{color}

            \DeclareFontFamily{U}{wncy}{}
            \DeclareFontShape{U}{wncy}{m}{n}{%
               <5>wncyr5%
               <6>wncyr6%
               <7>wncyr7%
               <8>wncyr8%
               <9>wncyr9%
               <10>wncyr10%
               <11>wncyr10%
               <12>wncyr6%
               <14>wncyr7%
               <17>wncyr8%
               <20>wncyr10%
               <25>wncyr10}{}

\newtheorem{thm}{Theorem}[section]

\newtheorem{lem}[thm]{Lemma}

\newtheorem{cor}[thm]{Corollary}
\newtheorem{prop}[thm]{Proposition}

\theoremstyle{definition}
\newtheorem*{dfn}{Definition}
\newtheorem{remark}{Remark}[section]

\newtheorem{example}{Example}[section]

\newcommand{\Z}{\mathbb Z}

\newcommand{\F}{\mathbb F}

\def\al{\alpha}

\def\s{\sigma}
\def\d{\delta}

\def\la{\lambda}

\def\be{\beta}

\begin{document}

\title[]{Matrix-Product Codes over Commutative Rings and Constructions Arising from $(\sigma,\delta)$-Codes}

\author{Mhammed Boulagouaz}
\address[M. Boulagouaz]{Department of Mathematics, Faculty of Sciences and Technologies, University of Sidi Mohamed Ben Abdellah, B.P. 2202, Fes, Morocco} \email{boulag@rocketmail.com}

\author{Abdulaziz Deajim}
\address[A. Deajim]{Department of Mathematics, King Khalid University,
P.O. Box 9004, Abha, Saudi Arabia} \email{deajim@kku.edu.sa, deajim@gmail.com}

\date{\today}

\begin {abstract}
A well-known lower bound (over finite fields and some special finite commutative rings) on the Hamming distance of a matrix-product code (MPC) is shown to remain valid over any commutative ring $R$. A sufficient condition is given, as well, for such a bound to be sharp. It is also shown that an MPC is free when its input codes are all free, in which case a generating matrix is given. If $R$ is finite, a sufficient condition is provided for the dual of an MPC to be an MPC, a generating matrix for such a dual is given, and characterizations of LCD, self-dual, and self-orthogonal MPCs are presented. Finally, results of this paper are used along with previous results of the authors to construct novel MPCs arising from $(\sigma, \delta)$-codes. Some properties of such constructions are also studied.
\end {abstract}

\maketitle

\section{{\bf Introduction}}
In this work, if not otherwise specified, $R$ denotes a  commutative ring with identity and $U(R)$ the multiplicative group of all invertible elements of $R$. A non-empty subset $C$ of the free $R$-module $R^n$ is called a code over $R$ of length $n$ and an element of $C$ is called a
codeword. If $C$ is an $R$-submodule of $R^n$, then $C$ is called a linear code over $R$. The $R$-submodule of $R^n$ generated by a code in $R^n$ is obviously a linear code over $R$, so all codes considered in this paper are linear. If $C$ is a free $R$-submodule of $R^n$ of rank $k$ (i.e. $C$ has an $R$-basis whose cardinality is $k$), then $C$ is called a free linear code over $R$ of rank $k$, and we express this by saying that $C$ is an $[n,k]$-linear code over $R$. If $C$ is an $[n,k]$-linear code over $R$, we say that a matrix $G \in M_{k, n}(R)$ is a  generating matrix of $C$ if the rows of $G$ form an $R$-basis of $C$. We, thus, write $C= \{x\,G\,|\, x\in R^{k}\}$.

On the free $R$-module $R^n$, consider the (Euclidean) bilinear form $<.\,,. > \,:R^n \times R^n \to R$ defined by $<(x_1, \dots, x_n),(y_1, \dots, y_n)>\,=\sum_{i=1}^n x_iy_i$. With respect to this bilinear form, define the dual $C^\perp$ of any code $C$ over $R$ by $C^\perp=\{x\in R^n\,|\, \mbox{$<x,y>\,=0$ for all $y\in C$}\}$. It is easy to check that $C^\perp$ is a linear code over $R$ if $C$ is. When $C \subseteq C^{\bot}$ (resp. $C = C^{\bot}$), we say that $C$ is self-orthogonal (resp. self-dual). A code $C$ is called linear complementary dual (LCD for short) if $C\cap C^\perp=\{0\}$ (see \cite{M}).

Let $C_i$ be an $[n,k_i]$-linear code over $R$, for $i = 1,\dots, s$. Writing codewords of the codes $C_i$ in column form, let $(c_1 \dots c_s)$ be the $n\times s$ matrix whose columns are $c_1 \in C_1,\dots, c_s\in C_s$. Consider the following subset of the set $M_{n\times s}(R)$ of $n\times s$ matrices with entries in $R$:
$$[C_1 \dots C_s] = \{(c_1 \dots c_s) \,|\, c_i\in C_i,\, i=1, \dots, s \}.$$
For $s\leq l$ and a matrix $A\in M_{s\times l}(R)$,
define the matrix-product code associated to
$C_{1}, \dots, C_{s}$ and $A$ to be
$$[C_1 \dots C_s]\, A:= \{(c_1 \dots c_s)A \,|\, c_i\in C_i,\, i=1, \dots, s \}\subseteq M_{n\times l}(R).$$
As the three $R$-modules $M_{n\times l}(R)$, $R^{nl}$, and $(R^n)^l$ are isomorphic, $[C_1 \dots C_s]\,A$ can be thought of as a code of length $nl$ over $R$ in an obvious way, and we can look at codewords of $[C_1 \dots C_s]\,A$ as elements of either of these three modules. More specifically, if $A=(a_{ij})$ and $c_i=(x_{1i},\dots,x_{ni})\in C_i$ for $i=1, \dots, s$, then, the codeword $(c_1 \dots c_s)\,A$ of $[C_1 \dots C_s]\,A$ is the matrix:
  $$\left(\begin{array}{ccc}
 \sum_{i=1}^s x_{1i} a_{i1} &  \cdots & \sum_{i=1}^s x_{1i} a_{il} \\
  \sum_{i=1}^s x_{2i} a_{i1} & \cdots & \sum_{i=1}^s x_{2i} a_{il} \\
  \vdots & \cdots & \vdots \\
  \sum_{i=1}^s x_{ni} a_{i1} & \cdots & \sum_{i=1}^s x_{ni} a_{il} \\
\end{array}\right)\in M_{n\times l}(R).$$
This matrix can be identified with the its corresponding element of $R^{nl}$, so the codeword $(c_1 \dots c_s)\,A$ can be looked at as the following element:
$$(\sum_{i=1}^s x_{1i} a_{i1}, \dots, \sum_{i=1}^s x_{1i} a_{il}, \dots, \sum_{i=1}^s x_{ni} a_{i1}, \dots, \sum_{i=1}^s x_{ni} a_{il})\in R^{nl}.$$
On the other hand, as the $k$th column of the above matrix is $\sum_{i=1}^s a_{ik}\,c_i\in R^{ n}$, the codeword $(c_1\dots c_s)\,A$ can be looked at as the following $l$-tuple with coordinates from $R^n$:
$$(\sum_{i=1}^s a_{i1}\,c_i, \dots, \sum_{i=1}^s a_{il}\,c_i)\in (R^n)^l.$$
It should be noted that $[C_1 \dots, C_s]\,A$ is linear if all $C_i$ are linear.

In general, some of the serious differences between linear codes over fields versus linear codes over commutative rings are apparent from the following:

1. A linear code $C\subseteq R^n$ may not be free.

2. Even when a code $C$ is free over $R$, its dual $C^\perp$ may not be free.

3. If $C$ and $C^\perp$  are both free codes over $R$ of length $n$, the equality $\mbox{rank}(C) +  \mbox{rank}(C^\perp)=n$ may not hold.

Due to point 1 above, it is not possible to talk about a generating matrix of a non-free code according to the definition of such a matrix we have given. With respect to points 2 and 3 above, it follows from \cite[Proposition 2.9]{FLL} that if $R $ is a finite commutative ring and $C$ is an $[n,k]$-linear code over $R$, then $C^\perp$ is an $[n,n-k]$-linear code over $R$. So, in most of our results, we work over finite commutative rings. Nonetheless, some results in this paper do not require finiteness of $R$, so we present them with no such restriction.

In Proposition \ref{rank}, we give sufficient conditions for a matrix-product code over a commutative ring to be free, and we give its generating matrix in Corollary \ref{COR}. We prove in Theorem \ref{INEQ} that a well-know lower bound for the minimum Hamming distance of a matrix-product code over a finite field or a finite chain ring remains valid over a commutative ring and we, further, give a sufficient condition for such a lower bound to be sharp. When we impose finiteness on $R$, more results are proved. Over such a ring, we generalize in Proposition \ref{P above} a well-know fact that tells when the dual of a matrix-product code is also a matrix-product code. This is used in Corollary \ref{dual 3} to give a generating matrix of the dual for a matrix-product code, and it is also used in Corollary \ref{self dual} to give characterizations of self-dual, self-orthogonal, and LCD matrix-product codes. As an interesting application, we study in Section \ref{sigma,delta} matrix-product codes arising from $(\s,\d)$-codes over finite commutative rings. In this section, we bring together results from the authors' work \cite{BD} and results proved in this paper to construct matrix-product codes out of $(\s, \d)$-codes, give generating matrices for such codes and their dual codes (Propositions \ref{star 1} and \ref{star 3}), and give a criterion in Proposition \ref{star 2} which tests when such a code is self-dual. Appropriate highlighting examples are also give.

\section{{\bf Matrix-Product Codes over Commutative Rings}}
In this section, unless further assumptions are stated, $R$ stands for a commutative ring with identity.

\begin{dfn} Let $A \in M_{s\times l}(R)$ with $s\leq l$.
\begin{itemize}
\item[(i)] If the rows of $A$ are linearly independent over $R$, we say that $A$  has full rank over $R$.

\item[(ii)] If there is $B\in M_{l \times s}(R)$ such that $AB = I_s$ (the $s\times s$ identity matrix), we say that $A$ is right-invertible and $B$ is the right inverse of $A$. Left-invertibility is defined similarly. If $s=l$, we say that $A$ is invertible if it has both right and left inverses.

\item[(iii)] If $s = l$ and the determinant $\mbox{det}(A)$ is a unit of $R$, then we say that $A$ is non-singular.
\end{itemize}
\end{dfn}

\begin{prop} If $A\in M_{s\times s}(R)$, then the following statements are equivalent:

(i) A is invertible.

(ii) A is non-singular.

If, further, $R$ is finite, then the above two statements are equivalent to the following:

(iii) A has full rank.
\end{prop}
\begin{proof}
The equivalence of the first two statements follows from the standard argument of computing the inverse of a square matrix (\cite{Mc}). For the last statement, see \cite[Corollary 2.8]{FLL}.
\end{proof}

\begin{prop} \label{rank} Let $A=(a_{ij})\in M_{s\times l}(R)$ be of full rank and $C_i$ be $[n, k_i]$-linear codes over $R$ for $i=1, \dots, s$. Then the matrix-product code $[C_1 \dots C_s]\,A$ is an $[nl, \sum_{i=1}^s k_i]$-linear code over $R$. If, further, $R$ is finite, then $\mid [C_1 \dots C_s]\,A\mid =\prod_{i=1}^{i=s }\mid C_i\mid $
\end{prop}
\begin{proof} The map
\begin{align*}
\phi: \prod_{i=1}^s C_i  \subseteq (R^n)^s &\rightarrow  [C_1 \dots C_s]\,A \subseteq M_{n\times l}(R)\simeq R^{nl},\\
 (c_1, \dots, c_s) &\mapsto(c_1 \dots c_s)\,A
\end{align*}
is a homomorphism of $R$-modules. If $(c_1, \dots ,c_s)\in \prod_{i=1}^s C_i$ is such that $\phi (c_1, \dots ,c_s)=0$, then
for each $ 1\leq k\leq n$ and $ 1\leq t\leq l$ we have  $\sum_{j=1}^{s}x_{kj}a_{jt}=0$. As $A$ has full rank over $R$, it follows that for each $ 1\leq k\leq n$ and $ 1\leq t\leq l$ we have $x_{kj}=0$. Therefore $\phi$ is injective. It is clear, by construction, that $\phi$ is surjective and, therefore, the rank of $[C_1 \dots C_s]\,A$ is equal to the rank of $\prod_{i=1}^s C_i$, which is $k_1+\dots+k_s $. Finally, the last statement follows from the bijectivity of $\phi$.
\end{proof}

\begin{lem} \label{FINITE} Let $M$ be a free $R$-module of rank $k$, and $S$ a system of $k$ elements of $M$. If $S$ generates $M$ over $R$, then $S$ is a basis of $M$ over $R$.
\end{lem}

\begin{proof} Assume that $S=(s_1,\dots ,s_k)$, and let $B=(e_1,\dots ,e_k)$ be an $R$-basis of $M$. Consider the $R$-module homomorphism
\begin{align*}
f: M=\sum_{i=1}^k Rs_i & \rightarrow  M=\bigoplus_{i=1}^k Re_i\\
 a_1s_1+\dots +a_k s_k &\mapsto a_1e_1+\dots +a_k e_k.
\end{align*}
If $\displaystyle\sum_{i=1}^k a_i s_i=0$ with $a_i\in R$, then $f(\sum_{i=1}^k a_i s_i)=\sum_{i=1}^k a_i e_i=0$. Since $e_1, \dots, e_s$ are linearly independent over $R$, $a_1=\dots =a_s=0$. Thus, $S$ is linearly independent and, hence, $S$ is an $R$-basis of $M$.
\end{proof}

\begin{remark}
In contrast with vector spaces over fields, one should be warned that with $M$ as in Lemma \ref{FINITE}, a linearly independent system whose cardinality is $k$ is not necessary an $R$-basis of $M$. For instance, looking at $\Z$ as a free $\Z$-module of rank 1, we notice that 2 is linearly independent over $\Z$ but does not generate $\Z$.
\end{remark}

\begin{cor}\label{COR}
Let $A=(a_{ij})\in M_{s\times l}(R)$ be of full rank, and $C_{1}, \cdots, C_{s}$ be $[n, k_i]$-linear codes over $R$, respectively. If $G_i\in M_{k_i \times n}(R)$ is a generating matrix of $C_i$ for $i=1, \dots, s$, respectively, then $[C_1 \dots C_s]\,A$ is free with the following generating matrix $(k=\sum_{i=1}^s k_i)$:
\[G=\left(%
\begin{array}{cccc}
  a_{11} G_1& a_{12}G_1 & \cdots & a_{1l}G_1 \\
  a_{21}G_2 & a_{22}G_2 & \cdots &a_{2l}G_2 \\
  \vdots & \vdots & \cdots & \vdots \\
  a_{s1} G_s&a_{s2}G_s & \cdots &a_{sl}G_s \\
\end{array}%
\right)\in M_{k, ln}(R).\]
\end{cor}
\begin{proof} By Proposition \ref{rank}, $[C_1 \dots C_s]\,A$ is free of rank $k=\sum_{i=1}^s k_i$ over $R$. The set $S$ consisting of the rows of the matrix $G$ is clearly a generating system of the code $[C_1 \dots C_s]\,A$ over $R$. Since the cardinality of $S$ is equal to the rank of $[C_1 \dots C_s]\,A$, it follows from Lemma \ref{FINITE} that $S$ is a basis of $[C_1 \dots C_s]\,A$ over $R$. Hence, $G$ is a generating matrix of $[C_1 \dots C_s]\,A$.
\end{proof}

Let $A=(a_{ij})\in M_{s\times l}(R)$. For $i = 1,\dots , s$, let $L_i = (a_{i1},\dots ,a_{il})\in R^l$ be the $i$th row of $A$ and $C_{L_i}$ the left $R$-submodule of $R^l$ generated by $L_1,\dots ,L_i$ (so, $C_{L_i}=\sum_{j=1}^{i} R\,L_j$). Let $d_i$ be the minimum Hamming distance of $C_i$ and $D_i$ the minimum Hamming distance of $C_{L_i}$. Generalizing its counterparts over a finite field (\cite{BN}) and a finite chain ring (\cite{As}), the theorem below gives a lower bound for the minimum Hamming distance of a matrix-product code $[C_1 \dots C_s]\,A$ over a commutative ring when $A$ has full rank. It, further, gives a sufficient condition under which the bound is sharp, generalizing \cite[Theorem 1]{HLR}. Note that, in the following theorem, we use the multiplication map $R^l \times R^n \to R^{ln}$ defined by:
\[UV=(u_1v_1,\dots,u_1v_n,u_2v_1,\dots,u_2v_n,\dots,u_lv_1,\dots, u_lv_n),\]
for $U=(u_1,\dots ,u_l)\in R^l$ and $V=(v_1,\dots, v_n)\in R^n$.

\begin{thm} \label{INEQ} Keep the above notation. If $A\in M_{s\times l}(R)$ is of full rank, then the minimum distance of the matrix-product code $[C_1 \dots C_s]\,A$ satisfies the following inequality:
$$d([C_1 \dots C_s]\,A) \geq \mbox{min}\{d_1D_1, d_2D_2,..., d_sD_s\}.$$
If, furthermore, $C_s \subseteq  C_{s-1}\subseteq \dots \subseteq C_1$ and, for every $i=1, \dots, s$, there exist $x_i\in C_i$ and $X_i\in C_{L_i}$ such that $\mbox{wt}(x_i)=d_i$, $\mbox{wt}(X_i)=D_i$, and $X_ix_i\neq 0$, then
   $$d([C_1 \dots C_s]\,A) = min\{d_1D_1,d_2D_2,\dots,d_sD_s\}.$$
\end{thm}

\begin{proof}
Let $0\neq c=(x_1 \dots x_s)\,A \in [C_1 \dots C_s]\,A$. There exists some $r\in \{1,\dots ,s-1\}$, $x_r \neq 0$ and $x_i=0$ for $r<i\leq s$ (So, $c=(x_1 \dots x_r 0 \dots 0)\,A$); otherwise set $r=s$. Since $0\neq x_r\in C_r$, $wt(x_r)\geq d_r$ and, thus, $x_r$ has at least $d_r$ nonzero components, $x_{i_1,r},...x_{i_{d_r},r}$ say. Now, for each $t =1,...,d_r$, we have $y_t=(x_{i_t, 1} x_{i_t, 2} \dots, x_{i_t, s})\,A\in C_{L_r}$ because $x_{j,k}=0$ for each $j=1,...,l$ and $r<k\leq s$. Since $x_{i_t,r}\neq 0$ and $A$ has a full rank over $R$, we deduce that $y_t\neq 0$. So, $wt(y_t)\geq D_r$. Hence,
\begin{align*}
wt(c) =wt((x_1x_2 \dots x_s)A) &=wt((x_{11}x_{12} \dots x_{1s})A,(x_{21}x_{22} \dots x_{2s})A,\dots ,(x_{n1}x_{n2} \dots x_{ns})A)\\
&= \sum_{k=1}^n \mbox{wt}((x_{k1}x_{k2} \dots x_{ks})A)\\
& \geq \sum_{t=1}^{ d_r}\mbox{wt}((x_{i_t,1} x_{i_t,2}\dots x_{i_t,s})A)\\
&= \sum_{t=1}^{d_r} \mbox{wt}(y_t)\\
& \geq d_rD_r\\
& \geq min\{d_1D_1,d_2D_2,\dots,d_sD_s\}.
\end{align*}

Now assume, further, that $C_s \subseteq  C_{s-1}\subseteq \dots \subseteq C_1$ and, for every $i=1, \dots, s$, there exist $x_i\in C_i$ and $X_i\in C_{L_i}$ such that $\mbox{wt}(x_i)=d_i$, $\mbox{wt}(X_i)=D_i$, and $X_ix_i\neq 0$. By the first part of this proof, we have $ d(C) \geq min\{d_1D_1,d_2D_2,\dots,d_sD_s\}$. Let $r\in \{1, \dots, s\}$ be such that $d_rD_r=min\{d_1D_1,d_2D_2,\dots,d_sD_s\}$. Take $x_r\in C_r$ and $X_r\in C_{L_r}$ so that $\mbox{wt}(x_r)=d_r$, $\mbox{wt}(X_r)=D_r$, and $x_r X_r \neq 0$. We show that $X_r x_r \in [C_1 \dots C_s]\,A$ and $\mbox{wt}(X_r x_r)=d_r D_r$, which settles the proof. Write $x_r=(x_{1,r}, \dots, x_{n,r})$ and $X_r=\sum_{i=1}^r \la_i L_i \in C_{L_r}$ for $\la_i\in R$. Set $y_i = \la_i x_r$ for $i=1, \dots, s$ with $\la_i=0$ for $i=r+1, \dots, s$. As $y_i\in C_r$, $y_i \in C_i$ for all $i=1, \dots, r$. Also, $y_i=0\in C_i$ for all $i=r+1, \dots, s$. Thus, $[y_1\dots y_s]\,A\in [C_1 \dots C_s]\,A$ and
\begin{align*}
[y_1\dots y_s]A&=(\sum_{i=1}^s y_i a_{i,1},\sum_{i=1}^s y_i a_{i,2},\dots,\sum_{i=1}^s y_i a_{i,l})\\&=(\sum_{i=1}^r \lambda_ia_{i,1},\sum_{i=1}^r\lambda_ia_{i,2},\dots,\sum_{i=1}^r \lambda_ia_{i,l})x_r\\ &=X_r x_r.
\end{align*}
As $0\neq X_r x_r\in [C_1 \dots C_s]\,A$, $\mbox{wt}(X_r x_r)\geq min\{d_1D_1,d_2D_2,\dots,d_sD_s\}= d_r D_r$. 
On the other hand, as precisely $d_r$ components of $x_r$ are nonzero and precisely $D_r$ components of $X_r$ are nonzero, it follows from the definition of the multiplication $X_r x_r$ that $\mbox{wt}(X_r x_r)\leq d_r D_r$. Hence, $\mbox{wt}(X_r x_r)=d_r D_r$ as claimed.
\end{proof}




\begin{remark} Note that if $R$ is a field (or even an integral domain), then the requirement on $x_i$ and $X_i$ in Theorem \ref{INEQ} holds automatically. On the other hand, we present here an example which shows that such a requirement is sufficient but not necessary. Let $R=\Z_4$ and consider the matrix $A=\left(\begin{array}{ccc} 1&2&0\\0&2&1 \end{array}\right)\in M_{2\times 3}(\Z_4)$. It is clear that $A$ is of full rank. Consider $C_{L_1}=\Z_4 (1,2,0)$ and $C_{L_2}=\Z_4 (1,2,0)+\Z_4(0,2,1)$. It can be checked that $D_1=D_2=1$, the only codewords in $C_{L_1}$ of weight 1 is $(2,0,0)$, and the only codewords in $C_{L_2}$ of weight 1 are $(2,0,0)$ and $(0,0,2)$. Set $C_1=C_{L_2}$ and $C_2=C_{L_1}$, so $C_2\subseteq C_1$ and $d^*=1$. It is easily seen that $x_1 X_1= x_2 X_1 = (0,0,0,0,0,0)$. Nonetheless, $d([C_1 C_2]\,A)=1=d^*$ since the codeword $$c=\left(\begin{array}{cc} 0& 2\\0&0\\0&0\\\end{array}\right)A=\left(\begin{array}{ccc} 0&0&2\\0&0&0\\0&0&0 \end{array}\right)\in [C_1 C_2]\,A$$
has weight 1.
\end{remark}

\section{{\bf On the dual of a matrix-product code over a finite commutative ring}}
Throughout this section, $R$ denotes a finite commutative ring with identity. A non-empty subset of the free $R$-module $M_{n\times m}(R)\cong R^{nm}$ can be looked at as a code over $R$ of length $nm$, where a codeword (which is a matrix $A \in M_{n \times m}(R)$) is thought of as a word over $R$ of length $nm$ in the obvious way. We consider the following bilinear form on $M_{n\times m}(R)$: $$<A,B> = \mbox{tr}(AB^T )=\sum_{i=1}^{n}\sum_{j=1}^{m} a_{ij}b_{ij},$$ for $A=(a_{ij})$ and $B=(b_{ij})$, where $B^T$ is the transpose of $B$ and $\mbox{tr}(AB^T)$ is the trace of the $n\times n$ matrix $AB^T$.

Our next goal is to give sufficient conditions for the dual of a matrix-product code to also be a matrix-product code, generalizing similar results over finite fields and finite chain rings (see for instance \cite{As} and \cite{BN}).

\begin{lem}\label{free dual}$($\cite[Proposition 2.9]{FLL}$)$
If $C$ is an $[n,k]$-linear code over $R$, then $C^\perp$ is an $[n, n-k]$-linear code over $R$.
\end{lem}

\begin{prop}\label{P above}
If $A\in M_{s\times s}(R)$ is non-singular,
and $C_1,\dots ,C_s$ are free linear codes over $R$ of length $n$, then $([C_1\dots C_s]\,A)^{\perp}=[C_1^{\perp}\dots C_s^{\perp}]\,(A^{-1})^{T}$.
 \end{prop}

\begin{proof}
Let $\mbox{rank}(C_j)=k_j$ for $j=1, \dots, s$. Since $R$ is finite, it follows from Lemma \ref{free dual} that $C_j^\perp$ are free over $R$ of rank $n-k_j$ for $j=1, \dots, s$. Thus, by Proposition \ref{rank}, both $[C_1 \dots C_s]\,A$ and $[C_1^\perp \dots C_s^\perp]\,(A^{-1})^T$ are free with
$$\mbox{rank}([C_1^\perp \dots C_s^\perp]\,(A^{-1})^T)=\sum_{j=1}^s (n-k_j)=ns-\sum_{j=1}^s k_j=ns-\mbox{rank}([C_1 \dots C_s]\,A).$$
From now on, just follow the proof of \cite[Lemma 6.1]{BN} with the obvious notational adjustments.
\end{proof}

\begin{remark}\label{Z_20}
Freeness of the input codes is necessary for the conclusion of Proposition \ref{P above} to hold, as the following example shows: Let $R=\Z_{20}$, $C_1=10\Z_{20}$, $C_2=2\Z_{20}$, and $A=\left(\begin{array}{cc} 3&0\\0&7\end{array}\right)$. It can be easily seen that $C_1^\perp=C_2$, $C_2^\perp=C_1$, and $A$ is non-singular with $(A^{-1})^T=\left(\begin{array}{cc} 7&0\\0&3\end{array}\right)$. Now, for $(8,2)\in [C_1 C_2]\,A$ and $(2,0)\in [C_1^\perp C_2^\perp]\,(A^{-1})^T$, it is clear that $<(8,2),(2,0)>=16 \neq 0$ in $\Z_{20}$. So, $([C_1 C_2]\,A)^\perp \neq [C_1^\perp C_2^\perp]\,(A^{-1})^T$. Notice that all assumptions of Proposition \ref{P above} are satisfied here except that $C_1$ and $C_2$ are not free over $\Z_{20}$.
\end{remark}

\begin{cor}\label{dual 3}
Keep the assumptions of Proposition \ref{P above}. If $G_1', \dots, G_s'$ are generating matrices of $C_1^\perp, \dots, C_s^\perp$, respectively, and $(A^{-1})^T=(b_{ij})$, then a generating matrix of $([C_1\dots C_s]\,A)^\perp$ is:
$$G'=\left(%
\begin{array}{cccc}
  b_{11} G_1'& b_{12}G_1' & \cdots & b_{1l}G_1' \\
  b_{21} G_2' & b_{22}G_2' & \cdots & b_{2l}G_2' \\
  \vdots & \vdots & \cdots & \vdots \\
  b_{s1} G_s'& b_{s2}G_s' & \cdots & b_{sl}G_s' \\
\end{array}%
\right).$$
\end{cor}

\begin{proof}
If $C_i$ is of rank $k_i$ over $R$ for $i=1, \dots, s$, it follows from Lemma \ref{free dual} that $C_i^\perp$ is free of rank $n-k_i$ for $i=1, \dots , s$. Letting $G_i'\in M_{(n-k_i),n}(R)$ be the respective generating matrices of $C_i^\perp$, the result now follows from Corollary \ref{COR} and Proposition \ref{P above}.
\end{proof}

For $r_1,\dots ,r_s\in R$, let $\mbox{Diag}(r_1,\dots ,r_s)$ denote the diagonal matrix of size $s\times s$ whose principal-diagonal entry in position $i,i$ is $r_i$ for $i=1, \dots, s$. 

 \begin{lem}\label{LDiag}  If $C_{1}, \dots, C_{s}$  are linear codes over $R$ of the same length, $Y\in M_{s\times l}(R)$, and $X=\mbox{Diag}(r_1,\dots,r_s)Y$ for some $r_1,\dots ,r_s\in U(R)$. Then, $[C_1 \dots C_s]\,X = [C_1 \dots C_s]\,Y$.
 \end{lem}
\begin{proof} As $X=\mbox{Diag}(r_1,\dots,r_s)Y$, $[C_1 \dots C_s]\,X=[(r_1C_1)\dots(r_sC_s)]\,Y$. Since $r_i$ is a unit, $r_iC_i=C_i$ for $i=1, \dots, s$, and thus the claimed conclusion follows immediately.
\end{proof}

The following result gives characterizations of self-dual, self-orthogonal, and LCD matrix-product codes over finite commutative rings.

\begin{cor}\label{self dual}
Let $A\in M_{s\times s}(R)$ be such that $AA^T=\mbox{Diag}(r_1, \dots, r_s)$ with $r_i\in U(R)$ for $i=1, \dots, s$, and let $C_1,\dots,C_s$ be linear codes over $R$ of the same length. Then

$1$. $[C_1 \dots C_s]\,A$ is self-dual if and only if
$C_i$ is self-dual for every $i=1,\dots ,s$.

$2$. $[C_1 \dots C_s]\,A$ is self-orthogonal if and only if
$C_i$ is self-orthogonal for every $i=1,\dots ,s$.

$3$. $[C_1 \dots C_s]\,A$ is LCD if and only if $C_i$ is LCD for every $i=1,\dots ,s$.
\end{cor}
\begin{proof} To begin with, as $\mbox{Diag}(r_1, \dots, r_s)$ is invertible and $A$ is a square matrix over a commutative ring (\cite{Mc}), $A$ and $A^T$ are invertible too, with $$(A^T)^{-1}=(A^{-1})^T=\mbox{Diag}(r_1^{-1}, \dots, r_s^{-1})A.$$ By Lemma \ref{LDiag},
\begin{align}
[C_1 \dots C_s]\, A=[C_1 \dots C_s]\,(A^{-1})^T.
\end{align}
Also, by Proposition \ref{P above},
\begin{align}
([C_1 \dots C_s]\,A)^\perp = [C_1^\perp \dots C_s^\perp]\,(A^{-1})^T.
\end{align}

1. By (1) and (2),
\begin{align*}
\mbox{$[C_1 \dots C_s]\,A$ is self-dual } & \mbox{iff }  [C_1 \dots C_s]\,A =([C_1 \dots C_s]\,A)^\perp\\
                                         & \mbox{iff }  [C_1 \dots C_s]\,(A^{-1})^T = [C_1^\perp \dots C_s^\perp]\,(A^{-1})^T\\
                                         & \mbox{iff }  \mbox{$C_i=C_i^\perp$ for every $1\leq i\leq s$ (as $A$ is invertible)}\\
                                         & \mbox{iff }  \mbox{$C_i$ is self-dual for every $1\leq i\leq s$}.
\end{align*}

2. Proving the self-orthogonality statement is similar.

3. We have
\begin{align*}
 \mbox{$[C_1\dots C_s]\,A$ is LCD } & \mbox{iff }  [C_1 \dots C_s]\,A \cap ([C_1 \dots C_s]\,A)^\perp=\{0_{R^{sn}}\}\\
                                     &\mbox{iff } [C_1\dots C_s]\,(A^{-1})^T \cap [C_1^\perp \dots C_s^\perp](A^{-1})^T=\{0_{R^{sn}}\}\\
                                        &\mbox{iff } \mbox{$[C_1 \dots C_s]\cap[C_1^\perp \dots C_s^\perp]=\{0_{R^{sn}}\}$ (as $A$ is invertible)}\\
                                        &\mbox{iff } \mbox{$C_i\cap C_i^\perp=\{0_{R^{n}}\}$ for every $i=1,\dots,s$}\\
                                        &\mbox{iff } \mbox{$C_i$ LCD for every $i=1,\dots,s$}.
\end{align*}
\end{proof}

\section{{\bf Matrix-product codes arising from $(\s,\d)$-codes over finite commutative rings}}\label{sigma,delta}
Throughout this section, $R$ denotes a finite commutative ring with identity. We use here some results from the authors' paper \cite{BD} combined with results from the previous sections to construct matrix-product codes based on $(\s, \d)$-codes over $R$ and, further, give a criterion for self-duality of such codes. We start off by recalling some terminologies and results from \cite{BD} and \cite{BL}.

\subsection{$(\s,\d)$-codes}\label{s,d codes}\hfill

For a ring endomorphism $\s$ of $R$ that maps the identity to itself and a $\s$-derivation $\d$ of $R$, let $R_{\s, \d}$ denote the (non-commutative) ring of skew-polynomials $\sum_{i=0}^m a_i X^i$ over $R$ with the usual addition of polynomials and multiplication based on the rule $Xa=\s(a)X+\d(a)$ for $a\in R$. Let $g(X)=\sum_{i=0}^{n-k}g_i X^i\in R_{\s, \d}$ be a monic skew-polynomial, for some $1\leq k \leq n-1$. Fix any monic $f(X)=\sum_{i=0}^n a_i X^i\in R_{\s, \d}$ of degree $n$ of which $g(X)$ is a right divisor in $R_{\s, \d}$, and let $(f)_l$ be the left principal ideal of $R_{\s, \d}$ generated by $f$. Then, $R_{\s, \d}/(f)_l$ is both a left $R_{\s, \d}$-module as well as a free left $R$-module with an $R$-basis $\mathcal{B}=\{1+(f)_l, X+(f)_l, \dots, X^{n-1}+(f)_l\}$. On the other hand, letting $$C_f =  \left( \begin{array}{ccccccc} 0 & 1 & 0 & 0 & \dots & 0 & 0\\ 0 & 0 & 1 & 0 & \dots & 0 & 0\\ \vdots & \vdots & \vdots & \vdots & \vdots & \vdots & \vdots\\ 0 & 0 & 0 & 0 & \dots & 0 & 1\\ -a_0 & -a_1 &  & \dots&  & -a_{n-2} & -a_{n-1} \end{array} \right)$$ be the companion matrix of $f$, define the group endomorphism $T_f:R^n\to R^n$ by $$(t_0, \dots, t_{n-1})\mapsto (\s(t_0), \dots, \s(t_{n-1}))C_f +(\d(t_0), \dots, \d(t_{n-1})).$$ Then the map $R_{\s, \d}\times R^n\to R^n$ given by $(P(X), (t_0, \dots, t_{n-1}))\mapsto P(T_f)(t_0, \dots, t_{n-1})$ defines a left action of $R_{\s, \d}$ on $R^n$ which makes $R^n$ a left $R_{\s, \d}$-module in an obvious way. Now, the map $\phi_f:R^n \to R_{\s, \d}/(f)_l$ given by $(t_0, \dots, t_{n-1})\mapsto \sum_{i=0}^{n-1} t_iX^i+(f)_l$ is a left $R_{\s, \d}$-module isomorphism. For every $p(X)\in R_{\s, \d}$, there is a unique $q(X)=\sum_{i=0}^{n-1} t_i X^i\in R_{\s, \d}$ (of degree at most $n-1$) such that $p(X)+ (f)_l=q(X)+(f)_l$. We can see that $(t_0, \dots, t_{n-1})=\phi_f^{-1}(p(X)+(f)_l)$, and we call $(t_0, \dots, t_{n-1})$ the coordinates of $p(X)+(f)_l$ with respect to the basis $\mathcal{B}$.

If $\mathcal{M}$ is a left $R_{\s, \d}$-submodule of $R_{\s, \d}/(f)_l$, the left $R_{\s, \d}$-submodule $C=\phi_f^{-1}(\mathcal{M})$ of $R^n$ is called an $(f, \s, \d)$-code (or just a $(\s, \d)$-code) of length $n$ over $R$. Note that $C$ consists of the coordinates of all the elements of $\mathcal{M}$. As $R$ is a subring of $R_{\s, \d}$, $\mathcal{M}$ and $C$ are also left $R$-modules. A linear code $C\subseteq R^n$ is called a principal $(f, \s, \d)$-code (or just a principal $(\s, \d)$-code) generated by $g$ if there exist monic skew-polynomials $f,g\in R_{\s, \d}$ of degrees $n$ and $n-k$, respectively, such that $g$ is a right divisor of $f$ in $R_{\s,\d}$ and $C=\phi_f^{-1}((g)_l/(f)_l)$. Such a code is free over $R$ of rank $k$ (see \cite[Theorem 1]{BL}). A $(\s, \d)$-code is called a principal $(\s, \d)$-constacyclic code if it is generated by some monic right divisor of $X^n-a$ for some $a\in U(R)$.

Starting with a set of monic skew-polynomials $g_1, \dots, g_s$ over $R$, we give here a construction of a free matrix-product code $C$ over $R$ whose input codes are principal $(\s,\d)$-codes generated by the $g_i$'s and, further, give its generating matrix in terms of the matrix of the code and the coefficients of the $g_j$'s. We also give a construction of the dual $C^\perp$ of $C$ under certain extra assumptions and give its generating matrix (a parity-check matrix of $C$).

For every $j=1, \dots, s$, let $\s_j$ be a ring endomorphism of $R$ that maps the identity to itself, $\d_j$ a $\s_j$-derivation of $R$, $g_j(X)=\sum_{i=0}^{n-k_j}g_{i,j}X^i\in R_{\s_j, \d_j}$ monic, and $C_j$ the principal $(\s_j, \d_j)$-code over $R$ generated by $g_j(X)$ (so, there exists a monic $f_j(X)\in R_{\s_j, \d_j}$ of degree $n$ of which $g_j(X)$ is a right divisor in $R_{\s_j, \d_j}$). By \cite[Theorem 2.7]{BD}, a generating matrix $G_j\in M_{k_j\times n}(R)$ of $C_j$ is given by
$$G_j=\left( \begin{array}{ccccccc} g_{0,j} & \dots  & g_{n-k_j,j} & 0& 0 &  \dots & 0\\ g_{0,j}^{(1)} & \dots & g_{n-k_j, j}^{(1)} & \s_j(g_{n-k_j, j}) & 0 &\dots & 0\\ g_{0,j}^{(2)} & \dots & g_{n-k_j,j}^{(2)} & g_{n-k_j+1,j}^{(2)} & \s_j^2(g_{n-k_j,j}) &\dots & 0\\  \vdots & \vdots & \vdots & \vdots  &\vdots& \vdots & \vdots \\ g_{0,j}^{(k_j-1)} & \dots & g_{n-k_j,j}^{(k_j-1)} & g_{n-k_j+1,j}^{(k_j-1)} & g_{n-k_j+2,j}^{(k_j-1)} &\dots & \s_j^{k_j-1}(g_{n-k_j,j}) \end{array} \right) ,$$
where

(1) $g_{i,j}=0$ for $n-k_j+1 \leq i \leq n-1$,

(2) $g_{0,j}^{(i)} =\d_j(g_{0,j}^{(i-1)})$ for $1\leq i \leq k_j-1$, and

(3) $g_{t,j}^{(i)}=\d_j(g_{t,j}^{(i-1)})+\s_j( g_{t-1,j}^{(i-1)})$ for $1\leq i \leq k_j-1$ and $1\leq  t\leq n-2$.\\

The matrices $G_j$ take more elegant shapes if $\d_j=0$ where, by \cite[Corollary 2.8]{BD}, we would have
$$G_j=\left( \begin{array}{cccccccc} g_{0,j}  & \cdots  & g_{n-k_j,j} & 0& 0 &  \cdots & 0\\ 0 & \s_j(g_{0,j})  & \cdots & \s_j(g_{n-k_j,j}) & 0 &\dots & 0\\ \vdots & \ddots  & &   &\ddots&  & \vdots \\ 0 & 0 & \cdots &0 &   \s_j^{k_j-1}(g_{0,j}) & \cdots & \s_j^{k_j-1}(g_{n-k_j,j}) \end{array} \right).$$

On the other hand, if further $\s_j$ are ring automorphisms of $R$ and $g_j(X)$ are also left divisors in $R_{\s_j, \d_j}$ of $f_j(X)=\sum_{i=0}^n a_{i,j} X^i$ for all $j=1, \dots, s$ with $f_j(X)=g_j(X)h_j(X)$ for $h_j(X)=\sum_{i=0}^{k_j} h_{i,j} X^i\in R_{\s_j, \d_j}$, then (by \cite[Theorem 5.1]{BD}) a generating matrix $H_j\in M_{n-k_j, n}(R)$ of $C_j^\perp$, for $j=1, \dots, s$, is
$$H_j=\left( \begin{array}{cccccccc} h_{k_j,j} & h_{k_j,j}^{(1)} & h_{k_j,j}^{(2)} & \dots & h_{k_j,j}^{(n-k_j-1)} & h_{k_j,j}^{(n-k_j)} & \dots & h_{k_j,j}^{(n-1)}\\ 0  & \s_j(h_{k_j,j}) & h_{k_j+1,j}^{(2)} & \dots & h_{k_j+1,j}^{(n-k_j-1)} & h_{k_j+1,j}^{(n-k_j)} & \dots & h_{k_j+1,j}^{(n-1)}\\ 0  & 0 & \s_j^2(h_{k_j,j}) & \dots & h_{k_j+2,j}^{(n-k_j-1)} & h_{k_j+2,j}^{(n-k_j)} & \dots & h_{k_j+2,j}^{(n-1)} \\ \vdots & \vdots & \vdots & \vdots & \vdots & \vdots & \vdots & \vdots \\ 0  & 0 & 0 & \dots & \s_j^{n-k_j-1}(h_{k_j,j}) & h_{n-1,j}^{(n-k_j)} & \dots & h_{n-1,j}^{(n-1)} \end{array} \right),$$
where

1) $h_{i,j}=0$ for $k_j+1 \leq i \leq n-1$,

2) for $1\leq i \leq n-k_j-1$ and $1 \leq t\leq n-2$,\\
\hspace*{1cm}(i) $h_{0,j}^{(i)} =\d_j(h_{0,j}^{(i-1)})$,\\
\hspace*{1cm}(ii) $h_{t,j}^{(i)}=\d_j(h_{t,j}^{(i-1)})+\s_j( h_{t-1,j}^{(i-1)})$,

3) for $n-k_j \leq i \leq n-1$ and $1 \leq t\leq n-1$\\
\hspace*{1cm}(i) $h_{0,j}^{(i)} =\d_j(h_{0,j}^{(i-1)})-a_{0,j}\s_j(h_{n-1,j}^{(i-1)})$, and\\
\hspace*{1cm}(ii) $h_{t,j}^{(i)}=\d_j(h_{t,j}^{(i-1)})+\s_j( h_{t-1,j}^{(i-1)})-a_{t,j}\s_j(h_{n-1,j}^{(i-1)})$.



\subsection{Matrix-product codes arising from $(\s,\d)$-codes}\label{s,d MPC}\hfill

Keep the notations and assumptions of Subsection \ref{s,d codes}. For $A\in M_{s\times l}(R)$, we denote the matrix-product code $[C_1 \dots C_s]\,A$ by $[g_1 \dots g_s]\,A$ in order to emphasise a way of constructing a free matrix-product over $R$ out of a well-chosen set of skew-polynomials over $R$, as the following results indicate.

\begin{prop}\label{star 1} Keep the notations and assumptions of Subsection \ref{s,d codes}. Let $A=(a_{ij})\in M_{s\times l}(R)$ be of full rank. Then, the matrix $G\in M_{r\times ln}(R)$ $(r=\sum_{i=1}^s k_i)$ given by
$$G= \left( \begin{array}{cccc}a_{11}G_1 &a_{12}G_1 & \dots &a_{1l}G_1  \\
a_{21}G_2 &a_{22}G_2 & \dots &a_{2l}G_2  \\
\vdots &\vdots & \cdots &\vdots  \\
a_{s1} G_s &a_{s2}G_s & \dots &a_{sl}G_s  \\
 \end{array} \right) $$
is a generating matrix for the matrix-product code $[g_1 \dots g_s]\,A$.
\end{prop}

\begin{proof} By \cite[Theorem 1]{BL} (see also \cite{BD}), $C_i$ is free of rank $k_i$ for every $i=1, \dots, s$. Now, applying Corollary \ref{COR} yields the claimed conclusion.
\end{proof}

\begin{prop}\label{star 3}
Besides the assumptions of Proposition \ref{star 1}, assume further that $\s_i$ is a ring automorphism of $R$, $g_i(X)$ is also a left divisor of $f_i(X)$ for $i=1, \dots, s$, and $A\in M_{s\times s}(R)$ is non-singular with $(A^{-1})^T=(b_{ij})$. Then the matrix $H\in M_{r' \times ln}(R)$ $(r'=\sum_{i=1}^s (n-k_i))$ given by
$$H= \left( \begin{array}{cccc}b_{11}H_1 &b_{12}H_1 & \dots &b_{1l}H_1  \\
b_{21}H_2 &b_{22}H_2 & \dots &b_{2l}H_2  \\
\vdots &\vdots & \cdots &\vdots  \\
b_{s1} H_s &b_{s2}H_s & \dots &b_{sl}H_s  \\
 \end{array} \right) $$
is a generating matrix for the dual matrix-product code $([g_1 \dots g_s]\,A)^\perp$.
\end{prop}

\begin{proof}
By the presentation in Subsection \ref{s,d codes}, $H_i$ is a generating matrix of $C_i^\perp$ for $i=1, \dots, s$. Now, apply Corollary \ref{dual 3} to get the conclusion.
\end{proof}

\begin{example}
Let $R$ be finite of characteristic 2, $\al,\be \in R$ with $\al^2+\al+1=0$, and $\s_1, \s_2$ ring automorphisms of $R$ with $\s_1(\al)=\al^2$ and $\s_2(\be)=\be$. We present several principal $\s_i$-codes of length 4 and use them to construct many matrix-product codes.\\
\indent \underline{Step 1:} Consider $f_1(X), g_1(X), h_1(X)\in R_{\s_1}$, with $f_1(X)=X^4+X^2+1$, $g_1(X)=X^2+\al$, and $h_1(X)=X^2+\al^2$. It can be checked that $$f_1(X)=g_1(X)h_1(X)=h_1(X)g_1(X).$$ Let $C_1$ and $C_2$ be the principal $\s_1$-codes of length 4 over $R$ generated, respectively, by $g_1$ and $h_1$. By Subsection \ref{s,d codes},
$$G_1=\left(\begin{array}{cccc} g_{0,1} &g_{1,1}& g_{2,1}&0\\0&\s_1(g_{0,1})&\s_1(g_{1,1})&\s_1(g_{2,1})\end{array}\right)=\left(\begin{array}{cccc} \al&0&1&0\\0&\al^2&0&1\end{array}\right)$$ is a generating matrix of $C_1$, and $$G_2=\left(\begin{array}{cccc} h_{0,1} &h_{1,1}&h_{2,1}&0\\ 0&\s_1(h_{0,1})&\s_1(h_{1,1})&\s_1(h_{2,1})\end{array}\right)=\left(\begin{array}{cccc} \al^2&0&1&0\\0&\al &0&1\end{array}\right)$$ is a generating matrix of $C_2$. On the other hand, $$H_1=\left(\begin{array}{cccc} h_{2,1}&h_{2,1}^{(1)}&h_{2,1}^{(2)}&h_{2,1}^{(3)}\\0&\s_1(h_{2,1})&h_{3,1}^{(2)}&h_{3,1}^{(3)}\end{array}\right)=\left(\begin{array}{cccc} 1&0&\al &0\\0&1&0&\al^2\end{array}\right)$$ is a generating matrix of $C_1^\perp$, and $$H_2=\left(\begin{array}{cccc} g_{2,1}&g_{2,1}^{(1)}&g_{2,1}^{(2)}&g_{2,1}^{(3)}\\0&\s_1(g_{2,1})&g_{3,1}^{(2)}&g_{3,1}^{(3)}\end{array}\right)=\left(\begin{array}{cccc} 1&0&\al^2 &0\\0&1&0&\al\end{array}\right)$$ is a generating matrix of $C_2^\perp$.
Notice that $G_1H_1^T=\left(\begin{array}{cc} 0&0\\0&0\end{array}\right)=G_2H_2^T$.\\
\indent \underline{Step 2:} Consider $f_2(X), g_2(X), h_2(X)\in R_{\s_2}$, with $f_2(X)=X^4+\be^4$, $g_2(X)=X+\be$, and $h_2(X)=X^3+\be X^2 +\be^2 X+\be^3$. It can be checked that $$f_2(X)=g_2(X)h_2(X)=h_2(X)g_2(X).$$ Let $C_3$ and $C_4$ be the principal $\s_2$-codes of length 4 over $R$ generated, respectively, by $g_2$ and $h_2$. By Subsection \ref{s,d codes},
$$G_3=\left(\begin{array}{cccc} g_{0,2} &g_{1,2}& 0 &0\\0&\s_2(g_{0,2})&\s_2(g_{1,2})&0\\ 0 & 0& \s_2^2(g_{0,2}) & \s_2^2(g_{1,2})   \end{array}\right)=\left(\begin{array}{cccc} \be&1&0&0\\0&\be&1&0\\ 0&0& \be & 1  \end{array}\right)$$ is a generating matrix of $C_3$, and $$G_4=\left(\begin{array}{cccc} h_{0,2} &h_{1,2}&h_{2,2}& h_{3,2}\end{array}\right)=\left(\begin{array}{cccc} \be^3&\be^2&\be&1 \end{array}\right)$$ is a generating matrix of $C_4$. On the other hand, $$H_3=\left(\begin{array}{cccc} h_{3,2}&h_{3,2}^{(1)}&h_{3,2}^{(2)}&h_{3,2}^{(3)}\end{array}\right)=\left(\begin{array}{cccc} 1&\be & \be^2 & \be^3\end{array}\right)$$ is a generating matrix of $C_3^\perp$, and $$H_4=\left(\begin{array}{cccc} g_{1,2}&g_{1,2}^{(1)}&g_{1,2}^{(2)}&g_{1,2}^{(3)}\\0&\s_2(g_{1,2})&g_{2,2}^{(2)}&g_{2,2}^{(3)}\\ 0&0& \s_2^2(g_1) & g_{3,2}^{(3)}\end{array}\right)=\left(\begin{array}{cccc} 1&\be &0&0\\0&1&\be &0\\ 0&0& 1& \be \end{array}\right)$$ is a generating matrix of $C_4^\perp$.
Notice that $G_3H_3^T=\left(\begin{array}{c} 0\\0\\0 \end{array}\right)$ and $G_4H_4^T=\left(\begin{array}{ccc} 0&0&0 \end{array}\right)$.\\
\indent \underline{Step 3:} Let $A=(a_{i,j})\in M_{2\times l}(R), B=(b_{i,j})\in M_{3\times l}(R), D=(d_{i,j})\in M_{4\times 3}(R)$ be all full-rank matrices. By Proposition \ref{star 1}, we can easily construct the generating matrices of many matrix-product codes out of different combinations of the above principal $\s_i$-codes such as $[g_i g_j]\,A$ and $[g_i h_j]\,A$, $[g_i h_j g_k]\,B$, $[h_i g_j h_k]\,B$, $[g_i h_j h_k g_p]\,D$, and $[g_i g_j g_k h_p]\,B$ for \linebreak $i,j,k,p=1,2$. For instance, $\left(\begin{array}{ccc} a_{1,1}G_1 &\dots& a_{1,l} G_1  \\ a_{2,1} G_4 &\dots & a_{2,l}G_4\end{array}\right)\in M_{3\times 4l}(R)$ is a generating matrix of $[g_1 h_2]\,A$, $\left(\begin{array}{ccc} b_{1,1} G_4 &\dots & b_{1,l}G_4\\ b_{2,1}G_3 &\dots & b_{2,l}G_3\\ b_{3,1}G_2 & \dots & b_{3,l}G_2\end{array}\right)\in M_{6\times 4l}(R)$ is a generating matrix of $[h_2 g_2 h_1]\,B$, and $\left(\begin{array}{ccc} d_{1,1} G_3 &\dots & d_{1,l}G_3\\ d_{2,1}G_4 &\dots & d_{2,l}G_4\\ d_{3,1} G_4 &\dots & d_{3,l}G_4\\ d_{4,1} G_1 &\dots & d_{4,l}G_1\end{array}\right)\in M_{7\times 4l}(R)$ is a generating matrix of $[g_2 h_2 h_2 g_1]\,D$. In a similar manner, we can construct generating matrices of different combinations of the codes $C_i$ and their dual codes for $i=1, 2,3,4$.\\
\indent \underline{Step 4:} Utlilizing Proposition \ref{P above}, we can give the generating matrices of the dual codes of all of the above matrix-product codes when the matrices $A$, $B$, and $D$ are square and non-singular. For instance, follwing Remark \ref{Z_20}, let $R$ be $\Z_{20}$ and $A=\left(\begin{array}{cc} 3&0\\0&7\end{array}\right)$. Then $A$ is non-singular and $(A^{-1})^T=\left(\begin{array}{cc} 7&0\\0&3\end{array}\right)$. As in Step 3, a generating matrix of $[g_1 g_2]\,A$ is $$G=\left(\begin{array}{cc} 3G_1&0\\0&7G_2\end{array}\right)\in M_{4\times 8}(R).$$ By Proposition \ref{P above}, a generating matrix of $([g_1 g_2]\,A)^\perp$ is
$$H=\left(\begin{array}{cc} 7H_1&0\\0&3H_2\end{array}\right)\in M_{4\times 8}(R).$$
Note that $$GH^T=\left(\begin{array}{cc} 21 G_1 H_1^T&0\\0&21 G_2H_2^T\end{array}\right)=\left(\begin{array}{cc} 0&0\\0&0\end{array}\right) \in M_{4\times 4}(R).$$
\end{example}
\vspace{0.5cm}

Besides the assumptions of Proposition \ref{star 3}, let us now assume further that, for every $j=1, \dots, s$, $\d_j=0$, $g_{0,j}\in U(R)$, $k=k_1=k_2=\dots=k_s$ with $n=2k$, and denote $R_{\s_j,0}$ by $R_{\s_j}$.

\begin{prop}\label{star 2} Keep the assumptions as above. Assume further that $A=(a_{ij})\in M_{s\times s}(R)$ is such that $AA^T=\mbox{Diag}(r_1, \dots, r_s)$ with $r_1, \dots, r_s\in U(R)$ and that, for every $j=1, \dots, s$, either of the following statements holds:
\begin{itemize}
\item[(1)] $g_j(X)$ is a right divisor in $R_{\s_j}$ of $X^n-a_j$ for some $a_j\in U(R)$, $C_j$ is the principal \linebreak $(X^n-a_j,\s_j)$-constacyclic code generated by $g_j(X)$, and $\s_j^{k_j}(h_{0,j}^{-1})h_j^*(X)=g_j(X)$, where \linebreak $g_j(X)h_j(X)=X^n-\s_j^{-k_j}(a_j)$.

\item[(2)] For any $l_j\in \{0,\dots, k_j\}$, $\sum_{i=0}^{l_j} \s_j^{k_j-1}(g_{i,j})g_{i+k_j-l_j,j}=0$.
\end{itemize}
Then, the matrix-product code $[g_1 \dots g_s]\,A$ is self-dual. 

\end{prop}

\begin{proof} By \cite[Corollary 3.7]{BD}, the statements (1) and (3) are equivalent and, further, they are equivalent to the condition of $C_j$ being self-dual. Now, apply Corollary \ref{COR} to get that $[g_1 \dots g_s]\,A$ is self-dual. 
\end{proof}

\begin{example} \label{EX}

Let $R=\F_3\times \F_3$, $\s(x,y)=(y,x)$, and $\al=(2,2)$. Note that $\s$ is a ring automorphism of $R$ of order 2, $1=1_R=(1,1)$, and $\al$ is a unit in $R$ of order 2. Set $\displaystyle{h(X)=X^2+\al X+\al}$.\linebreak Then $h_0=\al= h_0^{-1}$ and $h^*(X)= \al X^2+\al X+ 1$. Let $g(X)=\s^2(h_0^{-1})h^*(X)$. So, \linebreak $g(X)= \al(\al X^2+\al X+ 1)=X^2+X+\al$. Using the multiplication rule in $R_\s$, we have
\begin{align*}
g(X)h(X)&= (X^2+X+\al)(X^2+\al X +\al)\\
        &= X^4 +\s(\al)X^3+\s(\al)X^2+X^3+\s(\al)X^2+\s(\al)X+\al X^2 +\al^2 X +\al^2\\
        &= X^4 + (\al+1) X^3 + (3\al)X^2 +(\al+\al^2)X +1\\
        &=X^4 +1\\
        &=X^4+1
\end{align*}
Thus, $g(X)$ satisfies condition (2) of Proposition \ref{star 2} and, hence, it generates a self-dual $(X^4+1,\s)$-constacyclic code $C$ of length 4 over $R$. It follows from the paragraph following Proposition \ref{star 1} that a generating matrix of $C$ is $M = \left(\begin{array}{cccc} \al & 1 & 1 & 0 \\ 0 & \al & 1 & 1 \end{array}\right).$ Given the matrix $A= \left( \begin{array}{cc} (1,0) & (0,1)  \\ (0,2) &(1,0)  \\\end{array} \right)\in M_{2\times 2}(R)$, which is orthogonal and, thus, quasi-orthogonal, we conclude from Propositions \ref{star 2} and \ref{star 1}, respectively, that the matrix-product code $[g \, g]\,A$ is self-dual with the following generating matrix
 $$G=   \left( \begin{array}{cc}
 (1,0) M & (0,1)M  \\
 (0,2) M & (1,0)M  \\
 \end{array} \right)=
   \left( \begin{array}{cccccccc}( 2,0) &(1,0) & (1,0) & (0,0)&  (0,2) &(0,1) & (0,1) & (0,0)  \\
(0,0) &(2,0) &(1,0) & (1,0)  & (0,0) &(0,2) &(0,1) & (0,1)  \\

  (0,1 )&(0,2) &(0,2) & (0,0)& (2,0) &(1,0) &(1,0) & (0,0)  \\

 (0,0) &(0,1) &(0,2) & (0,2)& (0,0) &(2,0) &(1,0) & (1,0)   \\
 \end{array} \right).$$
\end{example}

\section*{Acknowledgement}
A. Deajim would like to express his gratitude to King Khalid University for providing administrative and technical support. He would also like to thank the University Council and the Scientific Council of King Khalid University for approving a sabbatical leave request for the academic year 2018-2019, during which this article was prepared.

\end{document}